\documentclass[a4paper]{article}
\pdfoutput=1
 
 \usepackage[margin=1.25in]{geometry}
\usepackage{verbatim}
\usepackage{amssymb,enumerate}
\usepackage{graphicx}
\usepackage{amsbsy}
\usepackage{authblk}
\usepackage{amsmath,amssymb,color}
\usepackage{graphicx}
\usepackage{amsmath,amsfonts}
\usepackage{enumerate}

\newtheorem{theorem}{Theorem}

\newtheorem{lemma}{Lemma}
\newenvironment{proof}{\paragraph{Proof:}}{\hfill$\square$}

\title{Relating Graph Thickness to Planar Layers and\\ Bend Complexity\footnote{A preliminary version appeared at the 43rd International Colloquium on Automata, Languages and Programming (ICALP 2016).}}

\author{Stephane Durocher\thanks{Work of the author is supported in part by the Natural Sciences and Engineering Research Council of Canada
 (NSERC).} } 
 
\author{Debajyoti Mondal}
\affil{Department of Computer Science, University of Manitoba, 
  Winnipeg,  Canada\\
  \texttt{\{durocher,jyoti\}@cs.umanitoba.ca}}

\begin{document}

\maketitle

\begin{abstract}
The thickness  of a graph $G=(V,E)$ with $n$ vertices
 is the minimum number of planar  subgraphs of $G$ whose union is $G$. A polyline drawing  of $G$ in $\mathbb{R}^2$ is a drawing $\Gamma$ of $G$, where each vertex is
 mapped  to a point and each edge is mapped to a polygonal chain. Bend
   and  layer complexities are two important aesthetics of such
 a drawing. The bend complexity of $\Gamma$ is the maximum number of bends
 per edge in $\Gamma$, and the  layer complexity of $\Gamma$ is the
 minimum integer $r$ such that the set of polygonal chains in $\Gamma$
 can be partitioned into $r$ disjoint sets, where each set corresponds to a planar polyline drawing. Let $G$ be a graph of thickness $t$.
  By F\'{a}ry's theorem, if $t=1$,
 then $G$ can be drawn on a single layer with bend complexity $0$.
 A few extensions to higher thickness are known, e.g., if $t=2$
  (resp., $t>2$), then $G$  can be drawn on
   $t$ layers with bend complexity 2
 (resp., $3n+O(1)$). However, allowing a higher number of layers
 may reduce the bend complexity, e.g., complete graphs require  $\Theta(n)$ layers to be drawn using 0 bends per edge.

In this paper we present an elegant extension  of F\'{a}ry's theorem to draw graphs of thickness $t>2$. We first prove that thickness-$t$ graphs    can be drawn on $t$ layers with $2.25n+O(1)$ bends per edge.  We then develop  another  technique to draw  thickness-$t$ graphs on $t$ layers with 
 bend complexity, i.e., $O(\sqrt{2}^{t} \cdot n^{1-(1/\beta)})$, where $\beta = 2^{\lceil (t-2)/2 \rceil }$. Previously, the bend complexity was not known to be sublinear for $t>2$.
 Finally, we  show that graphs  with linear  arboricity $k$ can be drawn on $k$ layers with bend complexity $\frac{3(k-1)n}{(4k-2)}$.    
\end{abstract}

\section{Introduction}
A polyline drawing of a graph $G=(V,E)$  in $\mathbb{R}^2$ maps each vertex of $G$
 to a distinct point, and each edge of $G$ to a polygonal chain.  	
 Many problems in VLSI layout and software visualization are tackled using
 algorithms that produce polyline drawings. For a variety of practical purposes, these
 algorithms often seek to produce drawings that  optimize
 several drawing aesthetics, e.g., minimizing the number of bends, minimizing the number of crossings, etc. 
 In this paper we examine two such parameters: \emph{bend complexity} and \emph{layer complexity}.

The \emph{thickness}  of a graph $G$ is the minimum number $\theta(G)$ such that $G$ can be decomposed  into $\theta(G)$ planar  subgraphs.  Let $\Gamma$ be a polyline drawing of $G$. Then 
 the \emph{bend complexity} of $\Gamma$ is the minimum integer $b$ such that
 each edge in $\Gamma$ has at most $b$ bends. 
 A set of edges $E' \subseteq E$ is called a \emph{crossing-free edge set in $\Gamma$}, if the corresponding polygonal chains correspond to a \emph{planar polyline drawing}, i.e., no two polylines that correspond to a pair of edges in $E'$ intersect, except possibly at their common endpoints.
 The \emph{layer complexity} of $\Gamma$ is the minimum integer $t$ such that 
 the edges of $\Gamma$ can be partitioned into $t$ crossing-free edge sets. 
 Figure~\ref{fig:intro}(a) illustrates a polyline drawing of $K_9$ on 3 layers with bend complexity 1. At first glance the layer complexity of $\Gamma$ may appear to be related to the thickness of $G$.
 However, the layer complexity is a property of the drawing $\Gamma$, while 
 thickness is a graph property. The layer complexity of $\Gamma$ can be arbitrarily large
 even when $G$ is planar, e.g., consider the case when
 $G$ is a matching and $\Gamma$ is a straight-line drawing, 
 where each edge crosses all the other edges; see Figure~\ref{fig:intro}(b). 

The layer complexity of a thickness-$t$ graph $G$ is at least $t$, and every $n$-vertex thickness-$t$
 graph admits a drawing on $t$ layers with bend complexity $O(n)$~\cite{PachW01}. 
 The problem of drawing thickness-$t$ graphs on $t$ planar layers is closely related to the \emph{simultaneous 
 embedding} problem, where given a set of planar graphs $G_1,\ldots,G_t$ on a common set of vertices, 
 the task is to compute their planar drawings $D_1,\ldots,D_t$ such that each vertex
 is mapped to the same point in the plane in each of these drawings. Figure~\ref{fig:intro}(a) can be thought as a simultaneous embedding of three given  planar graphs. 

\begin{figure}[pt]
\centering
\includegraphics[width=.95\textwidth]{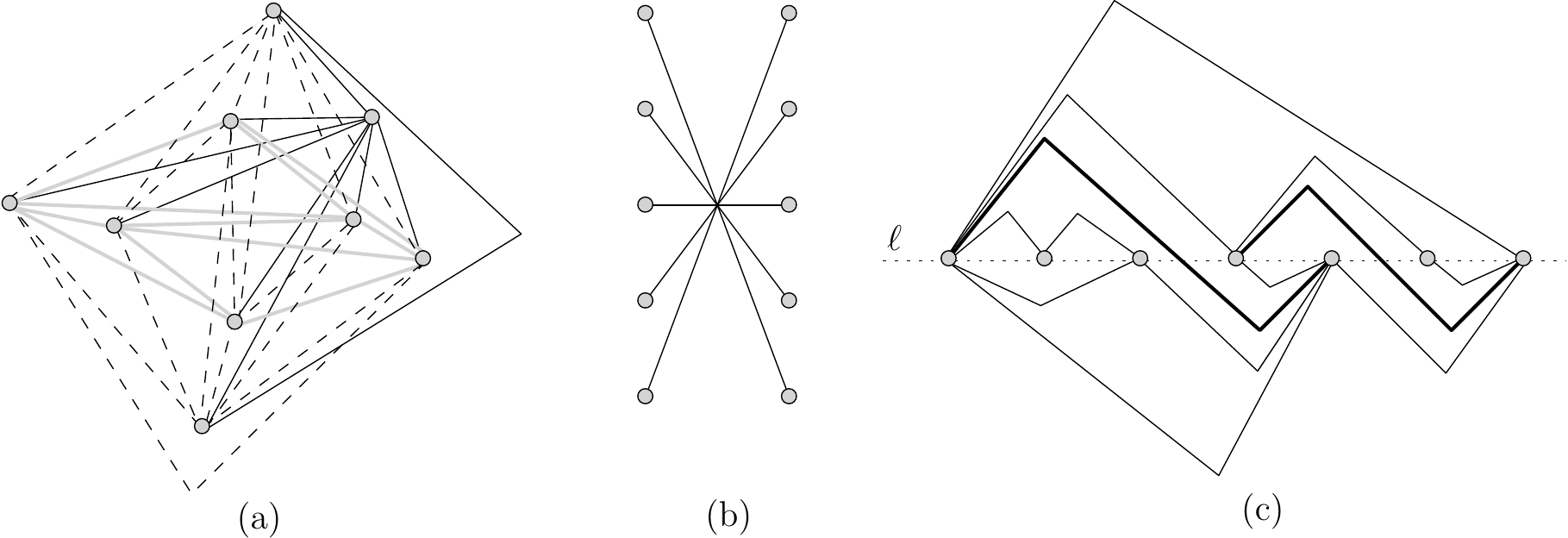}
\caption{(a) A polyline drawing of $K_9$. (b) A drawing of a matching of size  5. (c) A monotone topological book embedding of some graph. The edges that crosses the spine $\ell$ are shown in bold.}
\label{fig:intro}
\end{figure} 
 
\subsection{Related Work}

Graphs with low thickness admit polyline drawings on few layers with low bend complexity.
 If $\theta(G)=1$, then by F\'{a}ry's theorem~\cite{fary}, $G$ admits a drawing on a single layer with bend complexity $0$.
 Every pair of planar graphs can be simultaneously embedded using two bends per edge~\cite{EK05,GiacomoL07}.
 Therefore, if $\theta(G)=2$, then $G$ admits a drawing on two layers with bend complexity $2$. The best known lower bound on the bend complexity of such drawings  is one~\cite{DurocherGM13}. Duncan et al.~\cite{DuncanEK04} showed that graphs with maximum degree four can be drawn on two layers with bend complexity 0. Wood~\cite{Wood03} showed  how to construct drawings on  $O(\sqrt{m})$ layers  with bend complexity $1$, where $m$ is the number of edges in $G$.

Given an $n$-vertex planar graph $G$ and a point location for each vertex in $\mathbb{R}^2$,  Pach and Wenger~\cite{PachW01}
 showed that $G$ admits a planar polyline drawing with the given  vertex locations, where each edge has at most $120n$ bends. 
 They also showed that $\Omega(n)$ bends are sometimes necessary. 
 Badent et al.~\cite{BadentGL08} and Gordon~\cite{Gordon12} independently improved the bend complexity to $3n+O(1)$.
 Consequently, for $\theta(G)\ge 3$, these constructions can be used
 to draw $G$ on $\theta(G)$ layers with at most $3n+O(1)$ bends per edge. 

A rich body of literature~\cite{barat,HandbookSE,EnomotoM99,pachbook} examines \emph{geometric thickness}, i.e.,
 the maximum number of planar layers necessary to achieve 0 bend complexity. Dujmovi\'c and Wood~\cite{DujmovicW07} proved that  $\lceil k/2\rceil$ layers suffice 
 for graphs of treewidth $k$. Duncan~\cite{Duncan11} proved that 
 $O(\log n)$ layers suffice for graphs with arboricity two or outerthickness  two, and $O(\sqrt{n})$ layers
 suffice for  thickness-2 graphs. Dillencourt et al.~\cite{DillencourtEH00} proved that complete 
 graphs with $n$ vertices require at least  $\lceil (n/5.646)+0.342 \rceil$ and at most $\lceil    n/4 \rceil$ layers.

 
\subsection{Our Results} 
The goal of this paper is to extend our understanding of the 
 interplay between the layer complexity and bend complexity in polyline drawings.
 
 We first show that every $n$-vertex thickness-$t$ graph  admits a polyline drawing on $t$ layers  with bend complexity  $ 2.25n +O(1)$,  improving the  $3n+O(1)$ upper bound derived from~\cite{BadentGL08,Gordon12}.   
 We then give another drawing algorithm to draw thickness-$t$ graphs on $t$ layers with  bend complexity, i.e., $O(\sqrt{2}^{t} \cdot n^{1-(1/\beta)})$, where $\beta = 2^{\lceil (t-2)/2 \rceil }$. No such sublinear upper bound on the bend complexity was previously known for $t>2$. 
 Finally, we show that every $n$-vertex graph with linear arboricity $k\ge 2$ admits a polyline drawing on $k$ layers  with bend complexity $ \frac{3(k-1)n}{(4k-2)} $,   where the \emph{linear arboricity} of a graph $G$ is the minimum number of linear forests (i.e., each connected component is a path)  whose union   is $G$.

The rest of the paper is organized as follows. We start with some preliminary definitions and results (Section~\ref{sec:tech}). 
 In the subsequent section (Section~\ref{sec:algo}) we present two constructions to draw thickness $t$ graphs on $t$ layers. Section~\ref{sec:arbo} presents the results on drawing graphs of bounded arboricity. Finally, Section~\ref{sec:conclusion} concludes the paper pointing out the limitations of our results and suggesting directions for future research.

\section{Technical Details}
\label{sec:tech}
In this section we describe some preliminary definitions, and review some known results.

Let $G=(V,E)$ be a planar graph.  A \emph{monotone topological book embedding} of 
  $G$ is a planar drawing $\Gamma$ of $G$ that satisfies the 
 following properties. 
\begin{enumerate}
\item[P$_1$: ]  The vertices of $G$ lie along a horizontal line $\ell$ in  $\Gamma$. We refer to $\ell$ as the \emph{spine} of $\Gamma$.
\item[P$_2$: ]  Each edge $(u,v)\in E$ is an $x$-monotone polyline in $\Gamma$, where  
 $(u,v)$ either lies on one side of $\ell$, or crosses $\ell$ at most once.
\item[P$_3$: ]  Let $(u,v)$ be an edge that crosses $\ell$ at point $d$, where   $u$ appears before $v$ on $\ell$. Let 
 $u,\ldots,d,\ldots,v$ be the corresponding polyline.  Then 
 the polyline $u, \ldots,d$ lies above $\ell$, and 
 the polyline  $d,\ldots,v$ lies below $\ell$. 
\end{enumerate}
Figure~\ref{fig:intro}(c) illustrates a monotone topological book embedding of a planar graph.

 Let $G_1=(V,E_1)$ and $G_2=(V,E_2)$ be two graphs on a common set of vertices. 
 A \emph{simultaneous embedding} $\Gamma$ of $G_1$ and $G_2$ consists of their planar
 drawings $D_1$ and $D_2$, where each vertex is mapped to the same point in   
 the plane in both $D_1$ and $D_2$. Erten and Kobourov~\cite{EK05} showed that every
 pair of planar graphs admit a simultaneous embedding with at most three bends per edge. 
 Giacomo and Liotta~\cite{GiacomoL07} observed that by using monotone topological book embeddings  
 Erten and Kobourov's~\cite{EK05} construction can achieve a drawing with two bends per edge.
 Here we briefly recall this drawing algorithm. Without loss of generality 
 assume that both $G_1$ and $G_2$ are triangulations. Let $\pi_i$, where $1\le i \le 2$, be a vertex
 ordering that corresponds to a monotone topological book embedding of $G_i$. 
 Let $P_i$ be the corresponding \emph{spinal path}, i.e., a path that
  corresponds to $\pi_i$. Note  that some of the edges of $P_i$ may not
 exist in $G_i$, e.g., edges $(a,d)$ and $(b,c)$ in  Figures~\ref{fig:history}(a) and (b), respectively,
 and these edges of $P_i$ create  edge crossings in $G_i$.  
 Add a dummy vertex at each such edge crossing. Let $\delta_i(v)$  be the position of vertex $v$ in $\pi_i$. 
 Then $P_1$ and $P_2$ can be drawn simultaneously on an $O(n)\times O(n)$ grid~\cite{BrassCDEEIKLM07}
 by placing each vertex at the grid point $(\delta_1(v),\delta_2(v))$; see  Figure~\ref{fig:history}(c). The mapping between the dummy vertices of $P_1$ and $P_2$ can be arbitrary, here we map the dummy vertex on $(a,d)$ to the dummy vertex on $(b,c)$.   
 Finally, the edges of $G_i$ that do not belong to $P_i$ are drawn. Let $e$ be such an edge in $G_i$.
 If $e$ does not cross the spine, then it is drawn using one bend 
 on one side of $P_i$ according to the book embedding of $G_i$. Otherwise, let $q$ be a dummy vertex on the edge $e=(u,v)$, which corresponds to the intersection point of $e$ and the spine. The edges $(u,q)$ and $(v,q)$ are drawn on opposite sides of $P_i$ 
 such that the polyline from $u$ to $v$ do not create any bend at $q$. 
 Since each of $(u,q)$ and $(v,q)$  contains only one bend, $e$ contains only two bends.
 Finally, the edges of $P_i$  that do not belong to $G_i$ are removed from the drawing;
 see  Figure~\ref{fig:history}(d).

\begin{figure}[h]
\centering
\includegraphics[width=\textwidth]{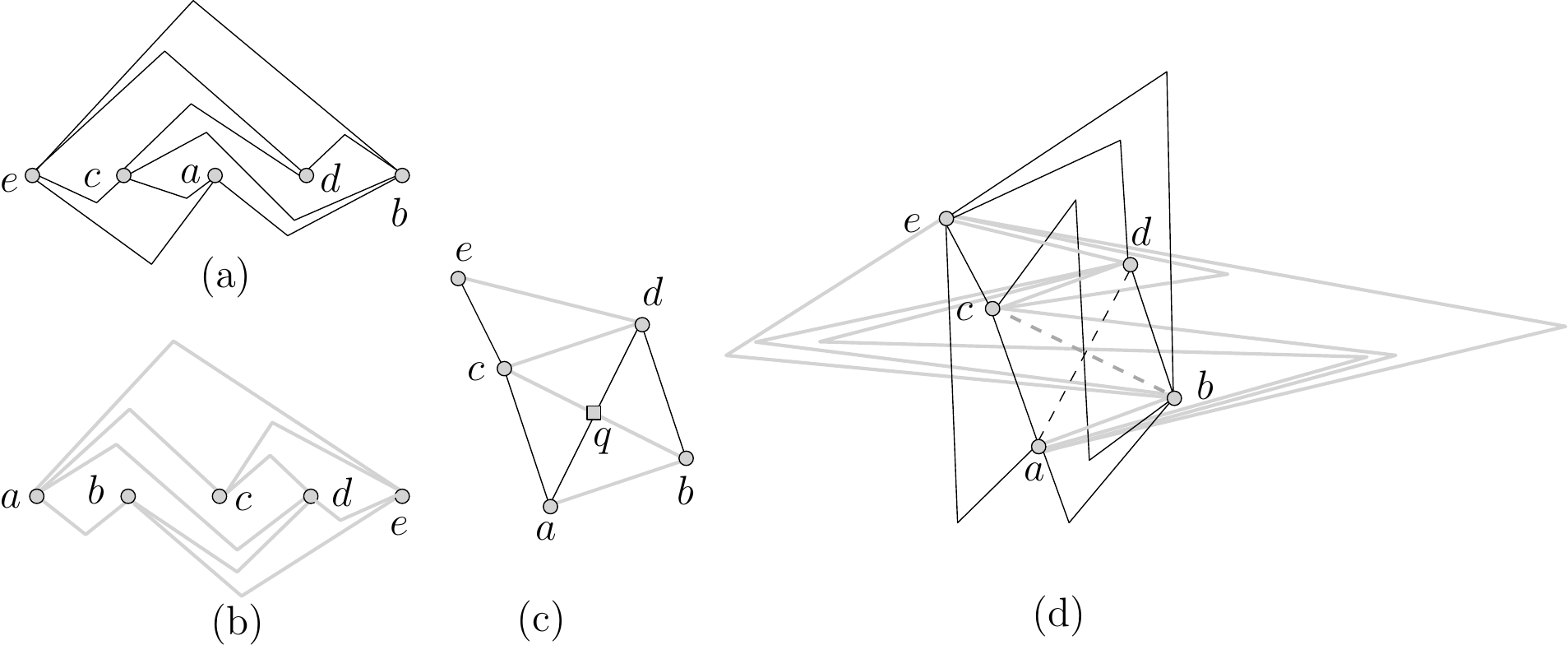}
\caption{(a)--(b) Monotone topological book embeddings of $G_1$ and $G_2$. 
 (c)--(d) Simultaneous embedding of $G_1$ and $G_2$, where the deleted edges are shown in dashed lines. }
\label{fig:history}
\end{figure} 

Let $\Gamma$ be a planar polyline drawing of a path $P=\{v_1,v_2,\ldots,v_n\}$. We call $\Gamma$ an \emph{uphill} drawing
 if for any point $q$ on $\Gamma$, the upward ray from $q$ does not intersect
 the path $v_1,\ldots,q$. Note that $q$ may be a vertex location or an interior point of some edge in $\Gamma$.
  Let $a$ and $b$ be two points in $\mathbb{R}^2$. 
  Then
 $a$ and $b$ are \emph{$r$-visible} to each other if and only if their exists a polygonal chain of length $r$ with end points
 $a,b$ that does not intersect $\Gamma$ at any point except possibly at $a,b$.
 A point \emph{$p$ lies between two other points $a,b$}, if either the inequality $x(a)<x(p)<x(b)$ or $x(b)<x(p)<x(a)$ holds.

 A set of points is \emph{monotone} if the 
 polyline connecting them from left to right is monotone
 with respect to  $y$-axis.  
 Let $S$ be a set of $n$ points in general position. 
 By the Erd\"os-Szekeres theorem~\cite{classic},  $S$ can be partitioned into $O(\sqrt{n})$ disjoint monotone subsets, and such a partition can be computed in $O(n^{1.5})$ time~\cite{Bar-YehudaF98}.

\section{Drawing Thickness-$\boldsymbol{t}$ Graphs on $\boldsymbol{t}$ Layers}
\label{sec:algo}
In this section we give two separate construction techniques to draw thickness-$t$ graphs on $t$ layers.  
 We first present  a construction  achieving   $2.25n +O(1)$ upper bound (Section~\ref{sec:tist}), which is  simple  and  intuitive.  
 Although the technique is simple, the idea of the construction will  be used frequently in the rest of the paper. Therefore, we explained the construction in reasonable details.  
 
 Later, we present a second  construction    (Section~\ref{sec:involved}), which is more involved, and relies on a deep understanding of the geometry of point sets. In this case, the upper bound on the bend complexity will depend  on some generalization of Erd\"os-Szekeres theorem~\cite{classic}, e.g., partitioning a point set into monotone subsequences in higher dimensions (Section~\ref{sec:tislarge}). 

\subsection{A Simple Construction with Bend Complexity $\boldsymbol{2.25n+O(1)}$}  
\label{sec:tist}
Let $G_1,\ldots,G_t$ be the planar subgraphs of the input graph  $G$,
 and let $S$ be an ordered set of $n$ points on a semicircular arc. 
 Let $V=\{v_1,v_2,\ldots,v_n\}$ be the set of vertices of $G$. We show that each  
 $G_i$, where $1\le i\le t$, admits a polyline drawing  with bend
 complexity $2.25n+O(1)$ such that  vertex $v_j$ is mapped to the $j$th point of $S$.
 To draw $G_i$, we will use the vertex ordering of its monotone topological book embedding.
 The following lemma will be useful to draw the spinal path $P_i$ of $G_i$.  

\begin{lemma}\label{lem:simple}
Let $S=\{p_0,p_1,\ldots,p_{n+1}\}$ be a set  of 
 points lying on an $x$-monotone semicircular arc (e.g., see Figure~\ref{fig:circ}(a)), and let 
 $P=\{v_1,v_2,\ldots,v_n\}$ be a path of $n$ vertices. 
 Assume that $p_0$ and $p_{n+1}$ are the leftmost and rightmost points of $S$, respectively, 
 and the points $p_1,\ldots,p_n$ are equally spaced between them in some arbitrary 
 order. Then $P$ admits an uphill drawing $\Gamma$ with the vertex $v_i$ 
 assigned to $p_i$, where $1\le i\le n$, and every point $p_i$ satisfies the following properties:
\begin{enumerate}
 \item[A.] Both the points $p_0$ and $p_{n+1}$ are  $(3n/4)$-visible to $p_i$.  
 \item[B.] One can draw an $x$-monotone polygonal chain from $p_0$ to $p_{n+1}$ with $3n/4$ bends that 
 intersects $\Gamma$ only at $p_i$.  
 \end{enumerate}
\end{lemma}
\begin{proof}
 We prove the lemma by constructing such a drawing $\Gamma$ for $P$. 
  The construction assigns a polyline for each edge of $P$. The resulting drawing
  may contain edge overlaps, and the bend complexity could be as large as $n-2$. 
  Later we remove these degeneracies and reduce the bend complexity to obtain $\Gamma$.

\subparagraph*{Drawings of Edges:} For each point $p_i\in S$, where $1\le i\le n$, we create
 an \emph{anchor point} $p'_i$ at $(x(p_i),y(p_i)+\epsilon)$, where 
 $\epsilon>0$. We choose $\epsilon$ small enough such that
 for any $j$, where $1 \le i\not=j \le n$, all the points of $S$   between $p_i$ and $p_j$ lie above $(p'_i,p'_j)$. Figure~\ref{fig:circ}(a) illustrates this property for 
 the anchor point $p'_1$.

 We first draw the edge $(v_1,v_2)$ using a straight line segment. 
 For each $j$ from $2$ to $n-1$, we now draw the edges $(v_j,v_{j+1})$ one after another.
 Assume without loss of generality that $x(p_j) < x(p_{j+1})$. 
 We call a point $p\in S$ between $p_j$ and $p_{j+1}$ 
 a \emph{visited point} if the corresponding vertex $v $ appears in $v_1,\ldots,v_j$, i.e.,  $v $ has already been placed at $p $. We draw an $x$-monotone polygonal   chain $L$ that starts at $v_j$, 
  connects the anchors of the intermediate visited points from left to right,
 and ends at $v_{j+1}$. Figure~\ref{fig:circ}(b) illustrates such a construction. 

 Since the number of bends on $L$ is equal to the number of 
 visited points of $S$ between $p_j$ and $p_{j+1}$, 
 each edge contains at most $\alpha$ bends, where $\alpha$
  is the number of points of $S$ between $p_j$ and $p_{j+1}$.
 
\begin{figure}[h]
\centering
\includegraphics[width=\textwidth]{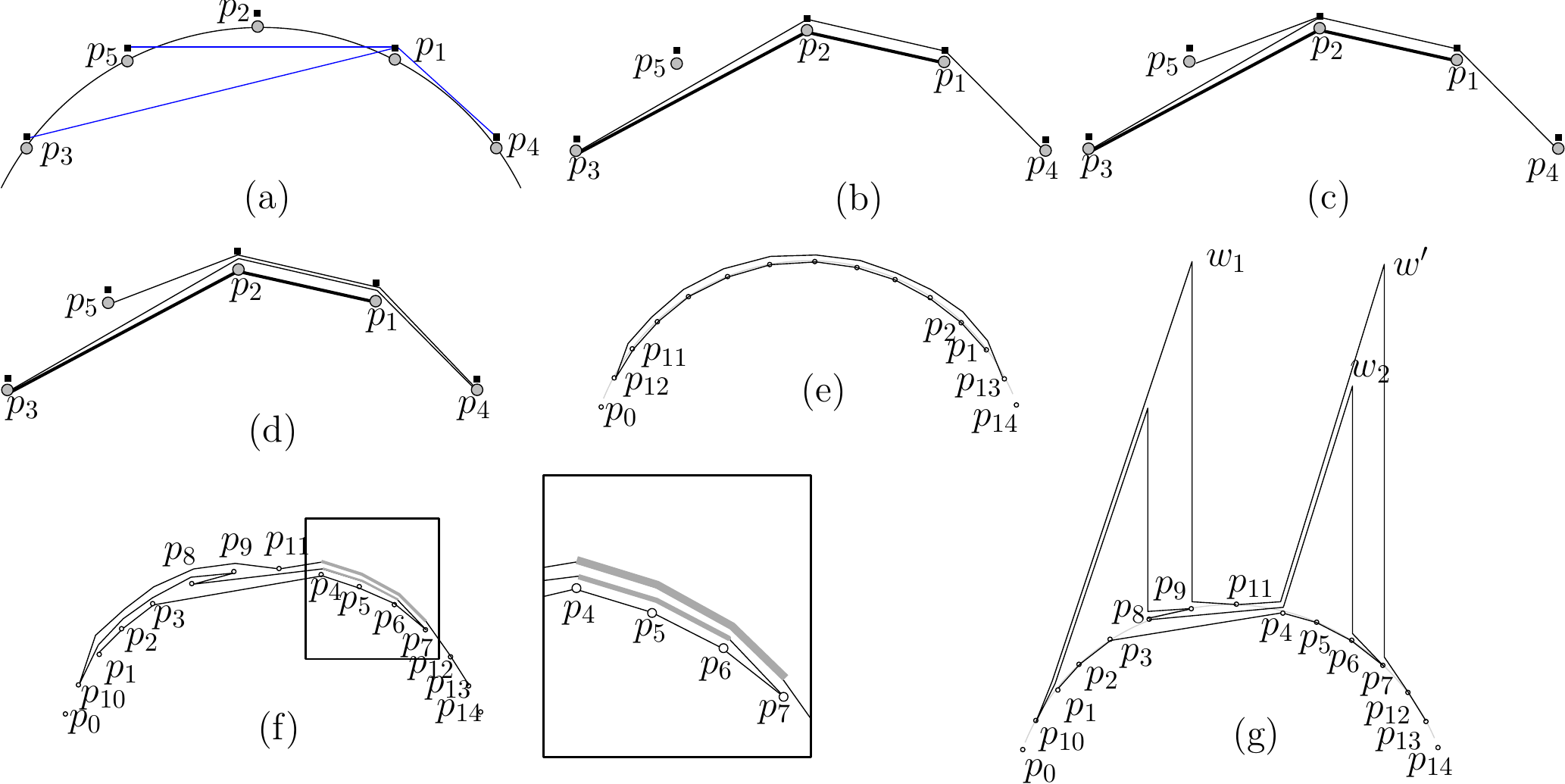}
\caption{Illustration for the proof of Lemma~\ref{lem:simple}. Anchor points are shown in black squares. For a larger view of this figure, see Appendix A.}
\label{fig:circ}
\end{figure} 

\textbf{Removing Degeneracies: }The drawing $D_n$ of the path $P$ constructed above contains edge 
 overlaps, e.g., see the edges $(v_3,v_4)$ and $(v_4,v_5)$ in Figure~\ref{fig:circ}(c). 
 To remove the degeneracies, for each $i$, we spread the corresponding bend points between
 $p_i$ and $p'_i$, in the order they appear on the path, see Figure~\ref{fig:circ}(d).
 Consequently, we obtain a planar drawing of $P$.   Let the resulting drawing be $D'_n$. Since each edge $(p_j,p_{j+1})$ is drawn as an $x$-monotone polyline above the path $p_1,\ldots,p_j$,  $D'_n$ satisfies the uphill property.  Note that $D'_n$ may have bend complexity $n-2$, e.g., see Figure~\ref{fig:circ}(e).
 We now show how to reduce the bend complexity and satisfy Properties A--B.
 
\textbf{Reducing Bend Complexity:} A pair of points in $S$ are \emph{consecutive} if they do not contain any other point of $S$
 in between. Let $e$ be any edge of $P$. Let $C_e$ be the corresponding
 polygonal chain in $D'_n$. A pair of bends on $C_e$ are called \emph{consecutive bends} if 
 their corresponding points in $S$ are also consecutive. 
 A  \emph{bend-interval} of $C_e$ is a maximal sequence of consecutive bends in $C_e$.  
 Note that we can partition the bends on $e$ into disjoint sets of bend-intervals. 
 
For any bend-interval $s$, let $l(s)$ and $r(s)$ be the $x$-coordinates of the left and 
 right endpoints of $s$, respectively.  
 Let $s_1$ and $s_2$ be two bend-intervals lying on two distinct edges $e_1$ and $e_2$ in $D'_n$, respectively,
 where $e_2$ appears after $e_1$ in $P$. 
 We   claim that    the  intervals $[l(s_1),r(s_1)]$ and $[l(s_2),r(s_2)]$ 
 are either disjoint, or $[l(s_1),r(s_1)] \subseteq [l(s_2),r(s_2)]$.
 We  refer to this property as the \emph{balanced parenthesis property of the bend-intervals}.  To verify this property assume that for some $s_1,s_2$, we have  $[l(s_1),r(s_1)] \cap [l(s_2),r(s_2)] \not= \phi$. Since  $s_2$ is a maximal sequence of consecutive bends, the inequalities $l(s_2)\le l(s_1)$ and  $r(s_2)\ge r(s_1)$ hold, i.e.,  $[l(s_1),r(s_1)] \subseteq [l(s_2),r(s_2)]$. We say that \emph{$s_1$ is nested by $s_2$}. Figure~\ref{fig:circ}(f) illustrates such a scenario, where $s_1,s_2$ are shown in thin and thick gray lines, respectively.
 
We now consider the edges of $P$ in reverse order, i.e., for each $j$ from 
 $n$ to 2, we modify the drawing of $e=(v_j,v_{j-1})$. For each bend-interval  $s=(b_1,b_2\ldots,b_r)$ of $C_e$, if $s$ has three or more bends, then we delete the bends $b_2,\ldots, b_{r-1}$, and join $b_1$ and $b_r$   using a new   bend point $w$. To create $w$, we consider the two cases of the balanced parenthesis property.  
 
 If $s$ is not nested by any other bend-interval in $D'_n$, then we place
 $w$ high enough above $b_r$ such that the chain $b_1,w,b_r$ does not introduce  any edge crossing, e.g., see the point $w_1(=w)$ in Figure~\ref{fig:circ}(g). On the other hand, if $s$ is nested by some other bend-interval,  then let $s'$ be such a bend-interval immediately
 above $s$. Since $s'=(b'_1,b'_2,\ldots,b'_r)$ is already processed, it must have been replaced by  some chain  $b'_1,w',b'_r$. Therefore, we can find a location for $b$ inside $\angle b'_1w'b'_r$
 such that the chain $b_1,w,b_r$ does not introduce  any edge crossing,   e.g., see the points $w'$ and $w_2(=w)$ in Figure~\ref{fig:circ}(g). 
 Let the resulting drawing of $P$ be $\Gamma$. 
 
We now show that the above modification reduces the bend complexity to $3n/4$. Let $e$ be an edge 
 of $P$ that contains $\alpha$ points from $S$ between its endpoints.
 Let $C_e$ be the corresponding polygonal chain in $D'_n$. Recall that 
 any bend-interval of length $\ell$ in $C_e$ contributes to $\min\{\ell ,3\}$ bends on $e$ in
 $\Gamma$. Therefore, if there are at most $\alpha /4$ bend-intervals on $C_e$,
 then $e$ can have at most $3\alpha /4$ bends in $\Gamma$.
 Otherwise, if there are more than $\alpha /4$ bend-intervals, then there are at least 
 $\alpha /4$ points\footnote{Every pair of consecutive bend-intervals contain such a point in between.} of $S$ that do not contribute to bends on $C_e$. Therefore, in both cases, 
 $C_e$ can have at most $3\alpha /4$  bends in $\Gamma$.

\textbf{Satisfying Properties A--B: }  Let $p_i$ be any point of $S\setminus \{p_0,p_{n+1}\}$. We first show that $p_0$  is $(3n/4)$-visible to $p_i$. Let $D_i$, where $1\le i\le n$,  be the drawing of the path $v_1,v_2,\ldots,v_i$. Observe that one can insert an edge 
 $(p_0,p_i)$ using an $x$-monotone polyline $L$ such that the bends on $L$ correspond to 
 the intermediate visited points. Now the drawing of the rest of the 
 path $v_i,v_{i+1}, \ldots, v_n$ can be continued such that it does not cross $L$.
 Therefore, if the number of points of $S$ between $p_0$ and $p_i$ is $\alpha$, then
 $L$ has at most $\alpha$ bends. Finally, the process of reducing bend complexity improves the number  of bends on $L$ to $3\alpha/4$. 
 
 Similarly, we can observe that $p_{n+1}$ is at most 
 $3\alpha'/4$ visible to $p_i$, where $\alpha'$ is the number of points of $S$ between $p_i$ and $p_{n+1}$.  Since the edges $(p_0,p_i)$ and $(p_i,p_{n+1})$ are $x$-monotone, we can draw an 
 $x$-monotone polygonal chain from $p_0$ to $p_{n+1}$ with at most $3(\alpha+\alpha')/4\le (3n/4)$ bends that  intersects $\Gamma$ only at $p_i$.
\end{proof}

\begin{theorem}
\label{thm:general}
Every $n$-vertex graph  of $t$ admits a drawing on $t$ layers with bend complexity $2.25n+O(1)$.
\end{theorem}
\begin{proof}
Let $G_1,\ldots,G_t$ be the planar subgraphs of the input graph  $G$,
 and let $V=\{v_1,v_2,\ldots,v_n\}$ be the set of vertices of $G$.
 let $S=\{p_0,p_1,\ldots,p_{n+1}\}$ be a set of $n+2$  
 points lying on a semicircular arc as defined in Lemma~\ref{lem:simple}. 
  Let $P_i$ be spinal path of the 
  monotone topological book embedding of $G_i$, where $1\le i\le t$.
  We first compute an uphill   drawing $\Gamma_i$ of the path $P_i$.
 We then draw the edges of $G_i$ that do not belong to $P_i$.
 Let $e=(u,v)$ be such an edge,  and without loss of generality assume that $u$ appears to the left of $v$ on the spine.
 
 If $e$ lies above (resp., below) the spine, then 
 we draw two $x$-monotone polygonal chains; one  from $u$ to $p_0$   (resp., $p_{n+1}$),
 and the other from $v$ to  $p_0$ (resp., $p_{n+1}$). By Lemma~\ref{lem:simple}, these polygonal chains 
 do not intersect $\Gamma_i$ except at $u$ and $v$, and  each contains at most $3n/4$ bends. Hence $e$ contains at most $1.5n$ bends in total.
  
 If $e$ crosses the spine, then it crosses some edge $(w,w')$ of $P_i$.
 Draw the edges $(u,w)$ and $(w,v)$ using the polylines $u,\ldots,p_0,\ldots,w$
 and   $w,\ldots,p_{n+1},\ldots,v$, respectively. The polylines $u,\ldots,p_0$
 and $p_{n+1},\ldots,v$ are $x$-monotone, and have at most $3n/4$ bends each. 
 The polyline $C=(p_0, \ldots,w\ldots, p_{n+1})$ is also $x$-monotone and 
 has at most $3n/4$ bends. Hence the number of bends is $2.25n$ in total. 
 It is straightforward to avoid the degeneracy at $w$, 
 by adding a constant number of bends on $C$.  
 
 Note that we still have some edge overlaps at $p_0$ and $p_{n+1}$. It is straightforward to remove
 these degeneracies by adding only a constant number of more bends   per edge.  
\end{proof}

\subsection{A Construction for Small Values of $\boldsymbol{t}$}
\label{sec:involved}
In this section we give another construction to draw thickness-$t$ graphs on $t$ layers. We first show that every thickness-$t$ graph, where $t\in\{3,4\}$,  can be drawn on $t$ layers with bend complexity $O(\sqrt{n})$, and then show how to extend the technique for larger values of $t$.
\subsubsection{Construction when $\boldsymbol{t=3}$}
\label{sec:tis3}

Let $S$ be an ordered set of $n$ points, where the ordering is by increasing $x$-coordinate.
 A \emph{$(k,n)$-group} $S_{k,n}$  is a partition of $S$
 into $k$ disjoint ordered subsets $\{S_1,\ldots,S_k\}$, each containing
 contiguous points from $S$. Label the points of $S$ using a permutation
 of $p_1,p_2,\ldots,p_n$ such that for each set $S'\in S_{k,n}$, 
 the indices of the points in $S'$ are either increasing or decreasing.
 If the indices are increasing (resp., decreasing), then we refer $S'$
 as a rightward (resp., leftward) set. We will refer to such a labelling as a
 \emph{smart labelling} of $S_{k,n}$.  
  Figure~\ref{fig:snake} illustrates a $(5,23)$-group and a smart labelling of the 
  underlying point set $S_{5,23}$.
 
 Note that for any $i$, where $1\le i\le n$, deletion of the points
 $p_1,\ldots,p_i$ removes the points of the rightward (resp., leftward) sets from their left (resp., right).
 The \emph{necklace} of $S_{k,n}$ is a path obtained from a smart labelling of $S_{k,n}$ by connecting the points $p_i,p_{i+1}$, where $1\le i\le n-1$.
 The following lemma constructs an uphill drawing of the necklace using $O(k)$ bends per edge.

\begin{figure}[pt]
\centering
\includegraphics[width=\textwidth]{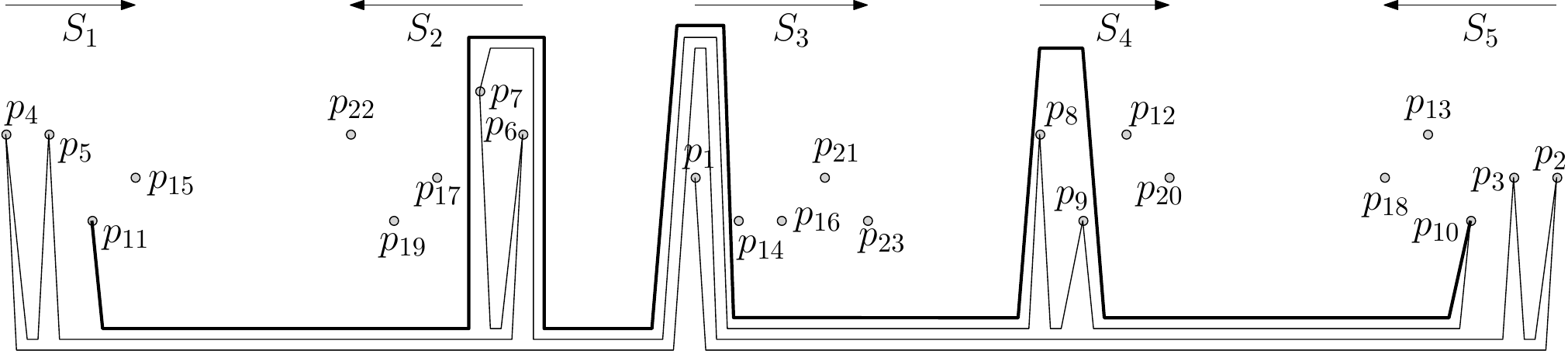}
\caption{Illustration for the proof of Lemma~\ref{lem:snake}. The edge $(p_{10},p_{11})$
 is shown in bold. Passing through each intermediate set requires at most 4 bends. }
\label{fig:snake}
\end{figure}

\begin{lemma}
\label{lem:snake}
 Let $S$ be a set of $n$ points ordered by increasing $x$-coordinate, and let $S_{k,n} = \{S_1,\ldots,S_k\}$  be 
 a $(k,n)$-group of $S$.  Label $S_{k,n}$ with a smart labelling. Then the necklace 
 of $S_{k,n}$ admits an uphill drawing with $O(k)$ bends per edge.
\end{lemma}
\begin{proof}
We construct this uphill drawing incrementally in a similar way as 
 in the proof of Lemma~\ref{lem:simple}. Let $D_j$, where $1\le j\le n$,
 be the drawing of the path $p_1,\ldots,p_j$. At each step of the construction, we maintain the 
 invariant that   $D_j$ is an uphill drawing. 
 
We first assign $v_1$ to $p_1$. Then for each $i$ from $1$ to $n-1$, we 
 draw the edge $(p_i,p_{i+1})$ using an $x$-monotone polyline $L$ that lies
 above $D_i$ and below the points $p_{j'}$, where $j'>i+1$. 
 Figure~\ref{fig:snake} illustrates such a drawing of $(p_i,p_{i+1})$.
 
The crux of the construction is that one can draw such a polyline $L$
 using at most $O(k)$ bends. Assume that $p_i$ and $p_{i+1}$
 belong to the sets $S_l\in S_{k,n}$ and $S_r\in S_{k,n}$, respectively.
 If $S_l$ and $S_r$ are identical, then $p_i$ and $p_{i+1}$ are consecutive,
 and hence it suffices to use at most $O(1)$ bends to draw $L$. On the other hand,
 if  $S_l$ and $S_r$ are distinct, then there can be at most $k-2$ sets of 
 $S_{k,n}$ between them. Let $S_m$ be such a set. While passing through $S_m$, we need to keep the points 
 that already belong to the path, below $L$, and the rest of the points above $L$.
 By the property of smart labelling, the points that belong to $D_i$
 are consecutive in $S_m$, and lie to the left or right side of $S_m$ depending on
 whether $S_m$ is rightward or leftward. Therefore, we need only $O(1)$
 bends to pass through $S_m$. Since there are at  most $k-2$ sets between 
 $S_l$ and $S_r$,  $O(k)$ bends suffice to construct $L$.  
\end{proof} 

We are now ready to describe the main construction. 
 Let $G$ be an $n$-vertex thickness-3 graph, and let $G_1,G_2,G_3$ be the planar subgraphs of $G$. 
 Let $P_i$ be the spinal path of the monotone topological
 book embedding of $G_i$, where $1\le i\le 3$.
 We first create a set  of $n$ points and assign them to the vertices of $G$.
 Later we route the edges of $G$.

\textbf{Creating Vertex Locations: } Assume without loss of generality that $P_1=(v_1,\ldots,v_n)$. 
 For each $i$ from $1$ to $n$, we place a point at $(i,j)$ in the plane, where  $j$ is 
 the position of $v_j$ in $P_2$. Let the resulting point set be $Q$.
 Recall that $Q$ can be partitioned into  disjoint monotone subsets $Q_1,\ldots,Q_k$,
 where $k\le O(\sqrt{n})$~\cite{Bar-YehudaF98}.
 Figure~\ref{fig:p}(a) illustrates such a partition. 
 
The sets $Q_1,\ldots,Q_k$ are ordered by the $x$-coordinate,
 and the indices of the labels of the points at each set
  is in increasing order. Therefore, if we place the points of the $i$th set between the lines  
 $x=2(i-1) n$ and $x=(2i-1)n$, then the resulting point set $Q'$
 would be a $(k,n)$-group, labelled by a smart labelling.
 Finally, we  adjust the $y$-coordinates of the points 
 according to the position of the corresponding vertices in $P_3$.
 Let the resulting point set be $S$.
 Figure~\ref{fig:p}(b) illustrates the vertex locations,
 where $P_1=(v_1,v_2,\ldots,v_n)$, $P_2=(v_{11},v_1,\ldots,v_3)$,
  and $P_3=(v_6,v_{11},\ldots,v_{10})$. 

\begin{figure}[pt]
\centering
\includegraphics[width=\textwidth]{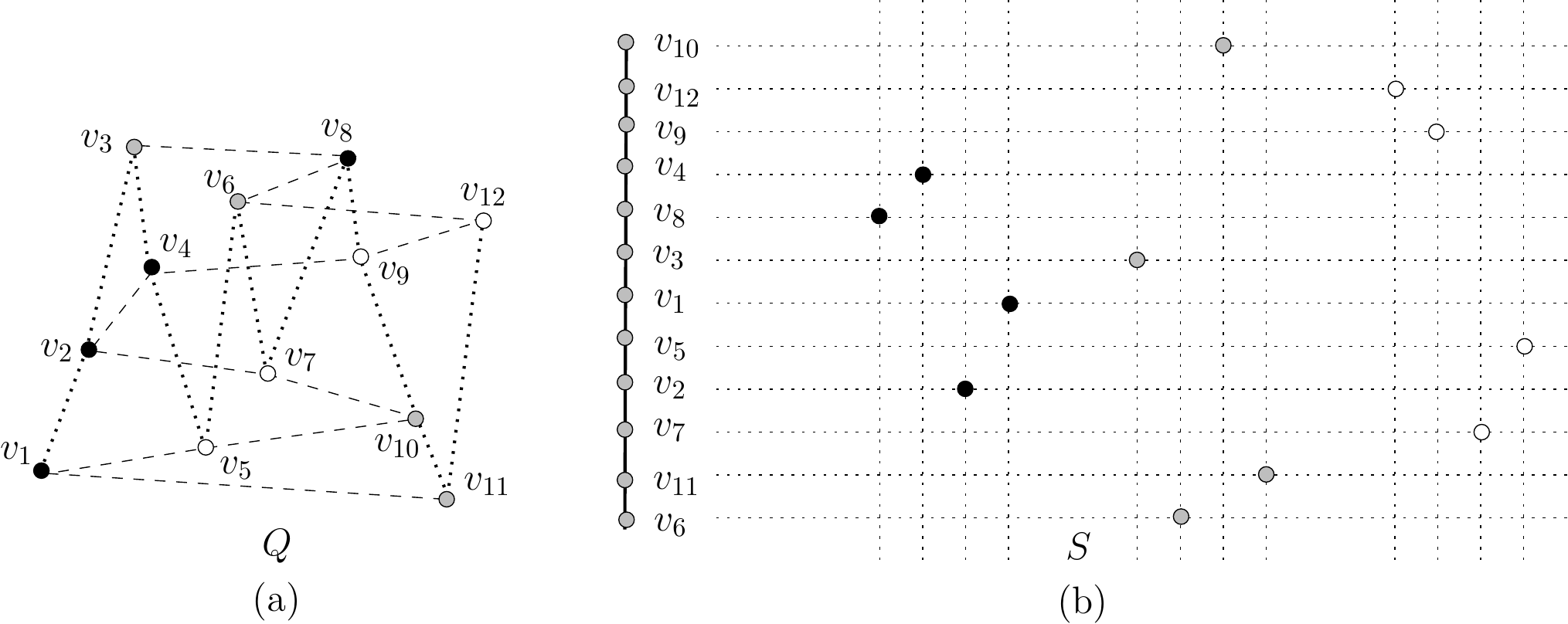}
\caption{Creating vertex locations for drawing thickness-3 graphs, where $P_1, P_2$ and $P_3$ are shown in dotted, dashed and thick solid lines, respectively.}
\label{fig:p}
\end{figure}

\textbf{Edge Routing: }It is straightforward to observe that the path $P_1$ is a 
 necklace for the current labelling of the points of $S_{k,n}$. 
 Therefore, by Lemma~\ref{lem:snake}, we can construct an uphill drawing of $P_1$ on $S$. 
 Observe that for every set $S'\in S_{k,n}$, the corresponding points are  monotone  in $Q$, i.e., the points of $S'$ are ordered along the $x$-axis  either in increasing or decreasing order of their $y$-coordinates in $Q$. Therefore,  relabelling the points according to the increasing order of 
 their $y$-coordinates in $Q$ will produce another smart labelling of $S$, and the corresponding necklace would be  the path $P_2$. Therefore, we can use 
 Lemma~\ref{lem:snake} to construct  an uphill drawing of $P_2$ on $S$.
 Since the height of the points of $S$ are adjusted according to the 
 vertex ordering on $P_3$, connecting the points  of $S$ 
 from top to bottom with straight line segments yields a $y$-monotone drawing of $P_3$.

We now route the edges of $G_i$ 
 that do not belong to $P_i$, where $1\le i\le 3$.
 Since $P_3$ is drawn as  a $y$-monotone polygonal path, 
 we can use the technique of Erten and Kobourov~\cite{EK05} to 
 draw the remaining  edges of $G_3$.
 To draw the edges of $G_2$, we insert two points $p_0$ and $p_{n+1}$ to the left and 
 right of all the points of $S$, respectively. Then the drawing of 
 the remaining edges of $G_1$ and $G_2$ is similar to the edge routing
 described in the proof of Theorem~\ref{thm:general}. 
 That is, if the edge $e=(u,v)$ lies above (resp., below) the spine, then
 we draw it using two $x$-monotone polygonal chains from
 $p_0$ (resp., $p_{n+1}$). Otherwise, if $e$ crosses the spine,
 then we draw three $x$-monotone polygonal chains, one from
 $u$ to $p_0$, another from $p_0$ to $p_{n+1}$, and the third one
 from $v$ to $p_{n+1}$. Since $k \le O(\sqrt{n})$, 
  the number of bends on
 $e$ is $O(\sqrt{n})$. Finally, we remove the degeneracies, which 
 increases the bends per edge by a small constant.

\subsubsection{Construction when $\boldsymbol{t=4}$} 
\label{sec:tis4}
We now show that the technique for drawing thickness-3 graphs can be 
 generalized to draw thickness-4 graphs with the same bend complexity. 
 
Let $G_1,\ldots,G_4$ be the planar subgraphs of $G$, and let $P_1,\ldots,P_4$ be the corresponding spinal paths. While constructing the vertex locations,
 we use a new $y$-coordinate assignment for the points of 
 $S$. Instead of placing the points according to the vertex ordering on the 
 path $P_3$, we create a particular order, by transposing the $x$- and $y$-axis, that would help to 
 construct uphill drawings of $P_3$ and $P_4$ with bend complexity $O(\sqrt{n})$.
 That is, we first create a $(k',n)$-group $S'_{k',n}$ using $P_3$ and $P_4$,  where $k'\in  O(\sqrt{n})$, in a similar 
 way that we created $S_{k,n}$ using $P_1$ and $P_2$. We then  
 adjust the $y$-coordinates of the points of $S$ according to the
 order these points appear in $S'_{k',n}$. Appendix B includes an example of such a construction. 
 
The construction of $G_1$ and $G_2$ remains the same as described in
 the previous section. However, since $P_3$ and $P_4$ now admit
 uphill drawings on $S$ with respect to $y$-axis, the drawing 
 of $G_3$ and $G_4$ are now analogous to the construction of 
 $G_1$ and $G_2$.

\subsubsection{Construction when $\boldsymbol{t>4}$} 
\label{sec:tislarge}
De Bruijn~\cite{classic2} observed that the result of Erd\"os-Szekeres~\cite{classic}    can be generalized to higher dimensions.
  Given a  sequence $\rho$ of $n$ tuples, each of size $\kappa$, one can find a subsequence of at least $n^{1/\lambda }$
  tuples, where $\lambda = 2^\kappa$,  such that they are monotone (i.e., increasing or decreasing) in every dimension. 
  This result is a repeated application of  Erd\"os-Szekeres result~\cite{classic} at each dimension. 
  We now show how to partition $\rho$ into few monotone sequences.
  
We use the partition algorithm of Bar{-}Yehuda and Sergio Fogel~\cite{Bar-YehudaF98} that   partitions a given 
 sequence of $n$ numbers into at most $2  \sqrt{n} $ monotone subsequences. 
 It is straightforward to restrict the size of the subsequences to $\sqrt{n}$, without increasing 
 the number of subsequences, i.e., by repeatedly extracting a monotone sequence of length exactly $\sqrt{n}$.
 Consequently, one can partition $\rho$ into 
 $2\sqrt{n}$ subsequences, where each subsequence is of length $\sqrt{n}$,
 and monotone  in the first dimension.
 By applying the partition algorithm on each of these 
 subsequences, we can find $2\sqrt{n} \cdot 2\sqrt{\sqrt{n}}$
 subsequence, each of which is of length $\sqrt{\sqrt{n}}$,
  and monotone in the first and second dimensions.
 Therefore, after $\kappa$ steps, we obtain a partition of $\rho$ into 
 $2^\kappa \cdot (n^{1/2} \cdot n^{1/4} \cdot \ldots \cdot n^{1/2^\kappa} = 2^\kappa \cdot n^{1-(1/\lambda)}$ monotone
 subsequences, where $\lambda = 2^\kappa$.
  We use this idea to extend our drawing algorithm to higher thickness.  
  
Let $G_1,\ldots,G_t$ be the planar subgraphs of $G$,  and let $P_1,\ldots,P_t$  be the corresponding 
 spinal paths. Let $v_1,v_2,\ldots,v_n$ be the vertices of $G$. Construct a corresponding sequence
 $\rho=(\tau_1,\tau_2,\ldots,\tau_n)$ of $n$ tuples, where each tuple is of size $t$, and the
 $i$th element of a tuple $\tau_j$ corresponds to the position of the corresponding vertex
 $v_j$ in $P_i$, where $1\le i\le t$ and $1\le j\le n$. We now partition  $\rho$ into a set
 of $2^t \cdot n^{1-(1/\beta)}$  monotone subsequences, where $\beta = 2^t$. 
 
 For each of these monotone sequences, we create an ordered set of consecutive points along the $x$-axis,
 where the vertex $v_j$ corresponds to the point $p_j$. It is now straightforward to observe that these 
 sets  correspond to a $(k,n)$-group $S_{k,n}$, where $k \le   2^t \cdot n^{1-(1/\beta)}$. Furthermore, since each group
 corresponds to a monotone sequence of tuples, for each $P_i$, the positions of the corresponding
 vertices are either increasing or decreasing. Hence, every path $P_i$ corresponds to a necklace
 for  some smart labelling of $S_{k,n}$. Therefore, by Lemma~\ref{lem:snake}, we can construct
 an uphill drawing of $P_i$ on $S$. We now add the remaining edges of $G_i$ following the 
 construction described in Section~\ref{sec:tis3}. Since $k \le 2^t \cdot n^{1-(1/\beta)}$, the number of bends is  
 bounded by $O(2^t \cdot n^{1-(1/\beta)})$.

Observe that all the points in the above construction have the same $y$-coordinate. Therefore,
 we can improve the construction by distributing the load equally among the $x$-axis and
 $y$-axis as we did in Section~\ref{sec:tis4}. Specifically, we 
 draw the graphs $G_1,\ldots,G_{\lceil t/2\rceil}$ using the uphill drawings of their
 spinal paths with respect to the $x$-axis, and the remaining graphs using the uphill
 drawings of their spinal paths with respect to the $y$-axis.      Consequently, the bend complexity decreases
 to $O(\sqrt{2}^t \cdot n^{1-(1/\beta')})$, where $\beta' = 2^{\lceil t/2 \rceil}$. 
 
 We can improve this bound further by observing that we are free to choose any arbitrary
 vertex labelling for $G$ while creating the initial sequence of tuples. Instead of using an arbitrary
 labelling, we could label the vertices according to their  ordering on some spinal path, which
 would reduce the bend complexity to  $O(\sqrt{2}^{t-2} \cdot n^{1-(1/\beta'')})$, where $\beta'' = 2^{\lceil (t-2)/2 \rceil }$.


\begin{theorem}
\label{thm:specific}
Every $n$-vertex graph $G$  of thickness $t\ge 3$ admits a drawing on $t$ layers with bend complexity
 $ O(\sqrt{2}^{t} \cdot n^{1-(1/\beta)})$, where $\beta = 2^{\lceil (t-2)/2 \rceil }$.
\end{theorem}

\section{Drawing Graphs of Linear Arboricity $\boldsymbol{k}$}
\label{sec:arbo}
In this section we construct polyline drawings, where the layer number and bend complexities are functions of the linear arboricity of the input graphs.     We show that the bandwidth of a   graph can be bounded in terms of its linear arboricity and the number of vertices, and then the result  follows from an application of  Lemma~\ref{lem:simple}.
 
 The \emph{bandwidth} of an $n$-vertex graph $G=(V,E)$ is
 the minimum integer $b$ such that the vertices can be labelled 
 using distinct integers from $1$ to $n$ satisfying the condition that for any edge $(u,v)\in E$,
 the absolute difference between the labels of $u$ and $v$ is at most $b$.
 The following lemma proves an upper bound on the  bandwidth of graphs.

\begin{lemma}  
\label{lem:band}
Given an $n$-vertex graph $G=(V,E)$ with linear arboricity $k$, the bandwidth of $G$ is at most $  \frac{3(k-1)n}{(4k-2)}$. 
\end{lemma}  
\begin{proof}
Without loss of generality assume that $G$ is a union of $k$ spanning paths $P_1,\ldots, P_k$. For any ordered sequence $\sigma$, let $\sigma(i)$ be the element at the $i$th position, and let $|\sigma|$ be the number of elements in $\sigma$. We now construct an ordered sequence $\sigma = \sigma_1 \circ  \sigma_2 \circ  \ldots \circ  \sigma_k\circ  \sigma_{k+1}$ of the vertices in $V$,  as follows. 

\begin{description}
\item[$\sigma_1$:] We initially place the first $x$  vertices of $P_1$ in the sequence, where the exact value of $x$ is to be determined later.
\item[$\sigma_2$:] We then place the vertices that are neighbors of $\sigma_1$ in $P_2$,
 in order, i.e., we first place the neighbors of $\sigma_1(1)$, then the neighbors of  $\sigma_1(2)$ that have not been placed yet, and so on.
\item[$\sigma_i$:] For each $i=3,\ldots,k$, we place the vertices that are
  neighbors of $\sigma_1$ in $P_i$ in order.
\item[$\sigma_{k+1}$:] We next place the remaining vertices of $P_1$ in order.
\end{description}

Figure~\ref{fig:arboricity}(a)  illustrates an example for three paths with $x=2$. Observe that $|\sigma_1|\le x$, and $|\sigma_t|\le 2x$, where $1< t\le k$. We now compute an upper bound on the bandwidth of $G$ using the vertex ordering of $\sigma$.
 
 For any $i,j$, where $1\le i <j\le k+1$, let  $\sigma_{i,j}$ be the sequence $ \sigma_i  \circ  \ldots \circ  \sigma_j$.  The edges of $P_1$ that are in $\sigma_1$ have bandwidth 1, and those that are in   $\sigma_1(x) \circ (\sigma \setminus \sigma_1)$ have bandwidth at most $(n-x)$,  e.g., see Figure~\ref{fig:arboricity}(b). 
  Now let $(v,w)$ be an edge of $G$ that does not belong to $P_1$. We compute the bandwidth of $(v,w)$  considering the following cases.

\begin{description}
\item[Case 1.] If none of $v$ and $w$ belongs to $\sigma_1$, then the bandwidth of $(v,w)$ is at most $(n-x)$.  
\item[Case 2.] If both $v$ and $w$ belong to $\sigma_1$, then the bandwidth of $(v,w)$ is at most $x$. 
\item[Case 3.] If at most one of $v$ and $w$ belongs to $\sigma_1$, then without loss of generality assume that $v$ belongs to $\sigma_1$. Since $(v,w)$ does not belong to $P_1$, we may assume that $w$ belongs to the path $P_t$, where $1<t\le k$. By the  construction of $\sigma$, $w$ belongs to $\sigma_{1,t}$, e.g., see Figure~\ref{fig:arboricity}(b). Without loss of generality assume that $w$ belongs to $\sigma_t'$, where $1< r \le t$. Let $u$ be the $q$th vertex in the sequence $\sigma$. Then the position of $w$ cannot be more than $q+2x\cdot (r-2) + 2q$, where the term $2x\cdot (r-2)$ corresponds to length of $\sigma_2\circ \ldots \circ\sigma_{r-1}$. Therefore, the bandwidth of the edge $(v,w)$ is at most $2x\cdot (r-2) + 2q\le 2x(r-1) \le 2x(t-1)$.
\end{description}

Observe that the bandwidth of the edges of $P_1$ is upper bounded by $(n-x)$. The bandwidth of any edge that belongs to $P_t$, where $1<t\le k$ is at most $2x(t-1)$. Consequently, the bandwidth of $G$ is at most $\max\{n-x, 2x(k-1)\} \le \frac{(2k-2)n}{(2k-1)}$, where $x= \frac{n}{(2k- 1)}$.
\end{proof}

\begin{figure}[h]
\centering
\includegraphics[width=\textwidth]{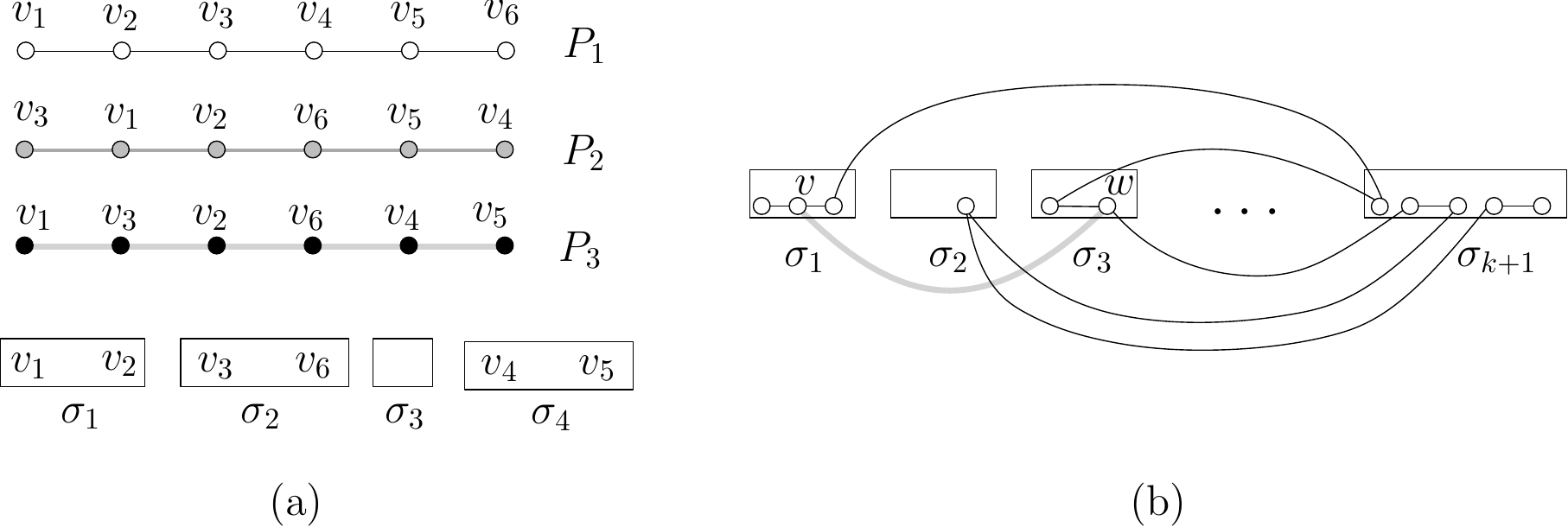}
\caption{(a) Construction of $\sigma$. (b) A schematic representation of $P_1$ and $(v,w)$, where $(v,w)$ belongs to $P_3$. }
\label{fig:arboricity}
\end{figure}

The following theorem is  immediate from the proof of  Lemmas~\ref{lem:simple} and~\ref{lem:band}.
\begin{theorem}
Every $n$-vertex graph with linear arboricity $k$ can be drawn on $k$ layers  with at most 
 $   \frac{3(k-1)n}{(4k-2)}  < 0.75n $ bends per edge.
\end{theorem}

\section{Conclusions}
\label{sec:conclusion}
 
In this paper we have developed algorithms to draw graphs on 
 few planar layers and with low   bend complexity.  
    Although our algorithms do not construct drawings with integral
    coordinates, it is straightforward to see that these drawings
    can also be constructed on polynomial-size integer grids,
    where all  vertices and bends have integral coordinates. We leave 
  the task of finding compact grid drawings achieving the same upper bounds  as a direction for  future research.   
  
We believe our upper bounds on bend complexity to be nearly tight,
 but we require more evidence to support this intuition. The only related 
 lower bound is that of Pach and Wenger~\cite{PachW01}, who
  showed that given a planar graph $G$ and a 
 unique location to place each vertex of $G$, $\Omega(n)$ bends are  sometimes 
  necessary to construct a planar polyline drawing of $G$ with the given vertex locations.    
  Therefore, a challenging research direction  would be to prove tight lower bounds on the bend complexity while drawing thickness-$t$ graphs on $t$ layers.

\bibliography{Thickness-Layer-Bends}

\begin{thebibliography}{10}

\bibitem{BadentGL08}
Melanie Badent, Emilio~Di Giacomo, and Giuseppe Liotta.
\newblock Drawing colored graphs on colored points.
\newblock {\em Theoretical Computer Science}, 408(2-3):129--142, 2008.

\bibitem{Bar-YehudaF98}
Reuven Bar{-}Yehuda and Sergio Fogel.
\newblock Partitioning a sequence into few monotone subsequences.
\newblock {\em Acta Informatica}, 35(5):421--440, 1998.

\bibitem{barat}
J\'anos Bar\'at, Ji\v{r}\'{i} Matou\v{s}ek, and David~R. Wood.
\newblock Bounded-degree graphs have arbitrarily large geometric thickness.
\newblock {\em Electronic Journal of Combinatorics}, 13(R3), 2006.

\bibitem{HandbookSE}
Thomas Bl\"asius, Stephen~G. Kobourov, and Ignaz Rutter.
\newblock Simultaneous embedding of planar graphs.
\newblock In Roberto Tamassia, editor, {\em Handbook of Graph Drawing and
  Visualization}, chapter~11, pages 349--380. CRC Press, August 2013.

\bibitem{BrassCDEEIKLM07}
Peter Bra{\ss}, Eowyn Cenek, Christian~A. Duncan, Alon Efrat, Cesim Erten, Dan
  Ismailescu, Stephen~G. Kobourov, Anna Lubiw, and Joseph S.~B. Mitchell.
\newblock On simultaneous planar graph embeddings.
\newblock {\em Computational Geometry}, 36(2):117--130, 2007.

\bibitem{DillencourtEH00}
Michael~B. Dillencourt, David Eppstein, and Daniel~S. Hirschberg.
\newblock Geometric thickness of complete graphs.
\newblock {\em Journal of Graph Algorithms and Applications}, 4(3):5--17, 2000.

\bibitem{DujmovicW07}
Vida Dujmovi\'c and David~R. Wood.
\newblock Graph treewidth and geometric thickness parameters.
\newblock {\em Discrete {\&} Computational Geometry}, 37(4):641--670, 2007.

\bibitem{Duncan11}
Christian~A. Duncan.
\newblock On graph thickness, geometric thickness, and separator theorems.
\newblock {\em Computational Geometry}, 44(2):95--99, 2011.

\bibitem{DuncanEK04}
Christian~A. Duncan, David Eppstein, and Stephen~G. Kobourov.
\newblock The geometric thickness of low degree graphs.
\newblock In {\em Proceedings of the 20th {ACM} Symposium on Computational
  Geometry (SoCG)}, pages 340--346. {ACM}, 2004.

\bibitem{DurocherGM13}
Stephane Durocher, Ellen Gethner, and Debajyoti Mondal.
\newblock Thickness and colorability of geometric graphs.
\newblock {\em Computational Geometry: Theory and Applications}, 56:1--18,
  2016.

\bibitem{EnomotoM99}
Hikoe Enomoto and Miki~Shimabara Miyauchi.
\newblock Embedding graphs into a three page book with ${O}(m \log n)$
  crossings of edges over the spine.
\newblock {\em {SIAM} Journal on Discrete Mathematics}, 12(3):337--341, 1999.

\bibitem{pachbook}
David Eppstein.
\newblock Separating thickness from geometric thickness.
\newblock In J\'anos Pach, editor, {\em Towards a Theory of Geometric Graphs}.
  American Mathematical Society, 2004.

\bibitem{classic}
Paul Erd\"os and George Szekeres.
\newblock A combinatorial theorem in geometry.
\newblock {\em Compositio Math.}, 2:463--470, 1935.

\bibitem{EK05}
Cesim Erten and Stephen~G. Kobourov.
\newblock Simultaneous embedding of planar graphs with few bends.
\newblock {\em Journal of Graph Algorithms and Applications}, 9(3):347--364,
  2005.

\bibitem{fary}
Istv\'{a}n F\'{a}ry.
\newblock On straight-line representation of planar graphs.
\newblock {\em Acta Sci. Math. (Szeged)}, 11:229--233, 1948.

\bibitem{GiacomoL07}
Emilio~Di Giacomo and Giuseppe Liotta.
\newblock Simultaneous embedding of outerplanar graphs, paths, and cycles.
\newblock {\em International Journal of Computational Geometry \&
  Applications}, 17(2):139--160, 2007.

\bibitem{Gordon12}
Taylor Gordon.
\newblock Simultaneous embeddings with vertices mapping to pre-specified
  points.
\newblock In {\em Proceedings of the 18th Annual International Conference on
  Computing and Combinatorics (COCOON)}, volume 7434 of {\em LNCS}, pages
  299--310. Springer, 2012.

\bibitem{classic2}
Joseph~B. Kruskal.
\newblock Monotonic subsequences.
\newblock {\em Proceedings of the American Mathematical Society}, 4:264--274,
  1953.

\bibitem{PachW01}
J\'anos Pach and Rephael Wenger.
\newblock Embedding planar graphs at fixed vertex locations.
\newblock {\em Graphs \& Combinatorics}, 17(4):717--728, 2001.

\bibitem{Wood03}
David~R. Wood.
\newblock Geometric thickness in a grid.
\newblock {\em Discrete Mathematics}, 273(1-3):221--234, 2003.

\end{thebibliography}

\begin{figure} [pt]
\section*{Appendix A: A Larger View of Figure~\ref{fig:circ}}\vspace{1cm}
\centering
\includegraphics[width=\textwidth]{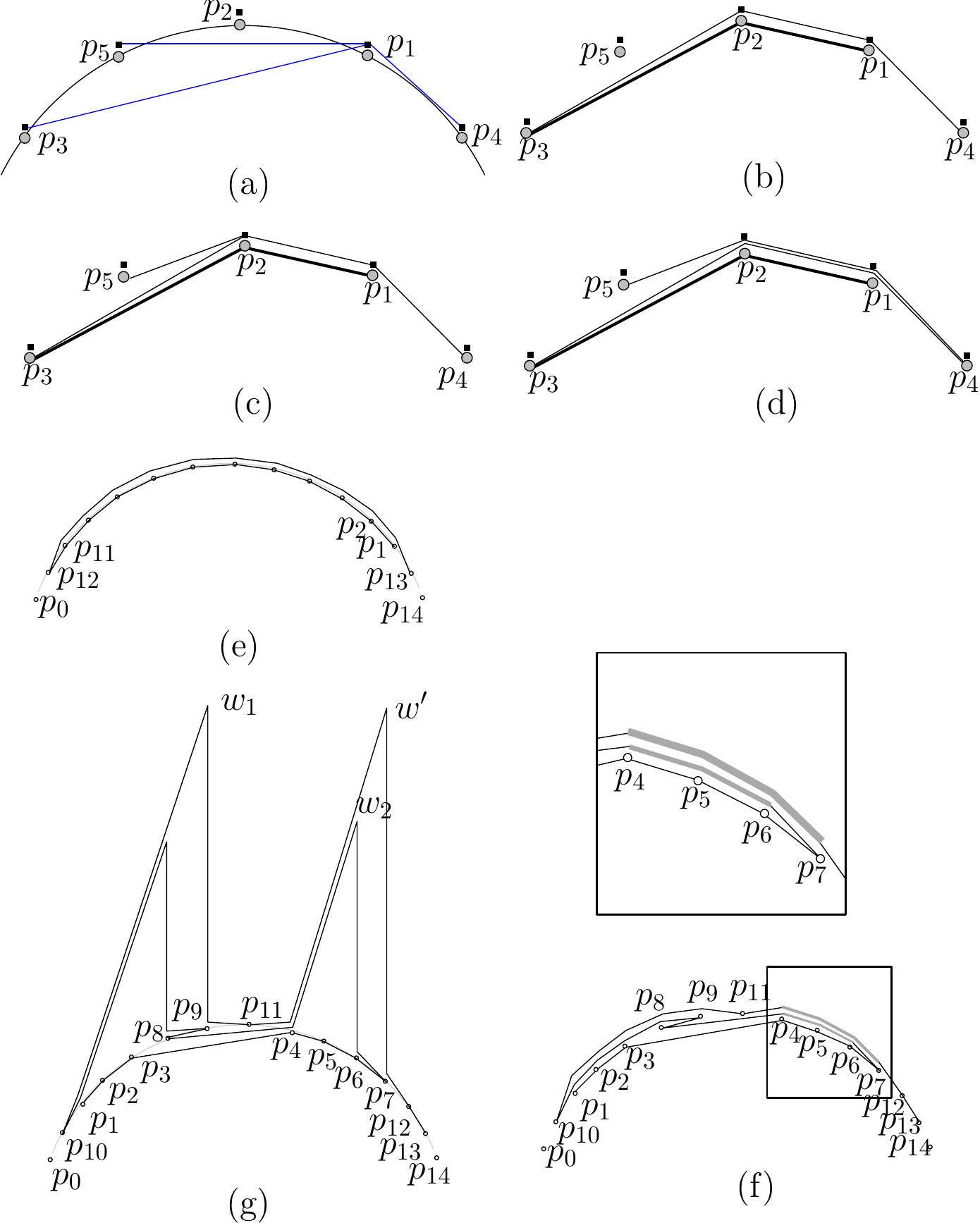}
\caption{Illustration for the proof of Lemma~\ref{lem:simple}. (a) Construction of the point set, and the anchor points. The anchor points are shown in black squares. (b)--(d) Construction of $D'_n$. (e) A scenario when the number of bends may be large. (f)--(g) Reducing bend complexity. 
}
\label{fig:circappendix}
\end{figure} 

\hfill
\newpage
 
\section*{Appendix B: Illustration for Drawing Thickness-4 Graphs} 

Here we illustrate the construction of the point set, as described in Section~\ref{sec:tis4}. Let $P_1,\ldots,P_4$ be the spinal paths of $G_1,\ldots,G_4$. 
 Figure~\ref{fig:1}(a) illustrates $P_1$ and $P_2$ in black and gray, respectively. Figure~\ref{fig:1}(b) illustrates $P_3$ and $P_4$ in black and gray, respectively. 
 
\begin{figure} [h]
\centering
\includegraphics[width=.9\textwidth]{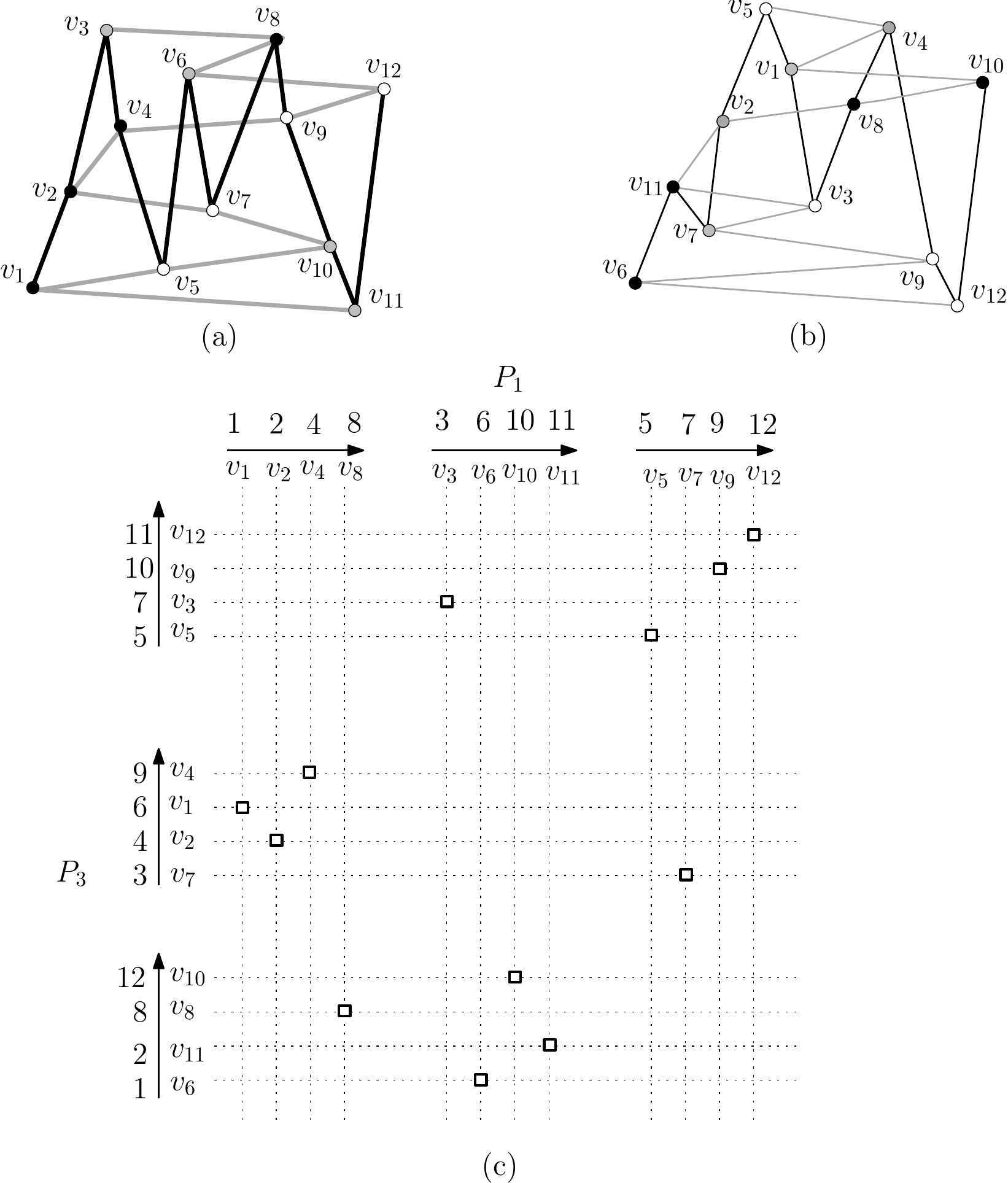}
\caption{(a) A point set, constructed from the paths $P_i$, where $i\in \{1,2\}$, by placing each vertex $v$ at $(\delta_1(v),\delta_2(v))$. Here $\delta_i(v)$ is the position of $v$ on $P_i$.  (b) A point set, constructed from the paths $P_i$, where $i\in \{3,4\}$, by placing each vertex $v$ at $(\delta_3(v),\delta_4(v))$.   (c) The final point set, and the corresponding $(k,n)$-groups. The numbers denote    the vertex positions on the corresponding spinal path. The arrows  illustrate whether the corresponding sets are leftward or rightward. 
}
\label{fig:1}
\end{figure}

\begin{figure} [pt]
\centering
\includegraphics[width=\textwidth]{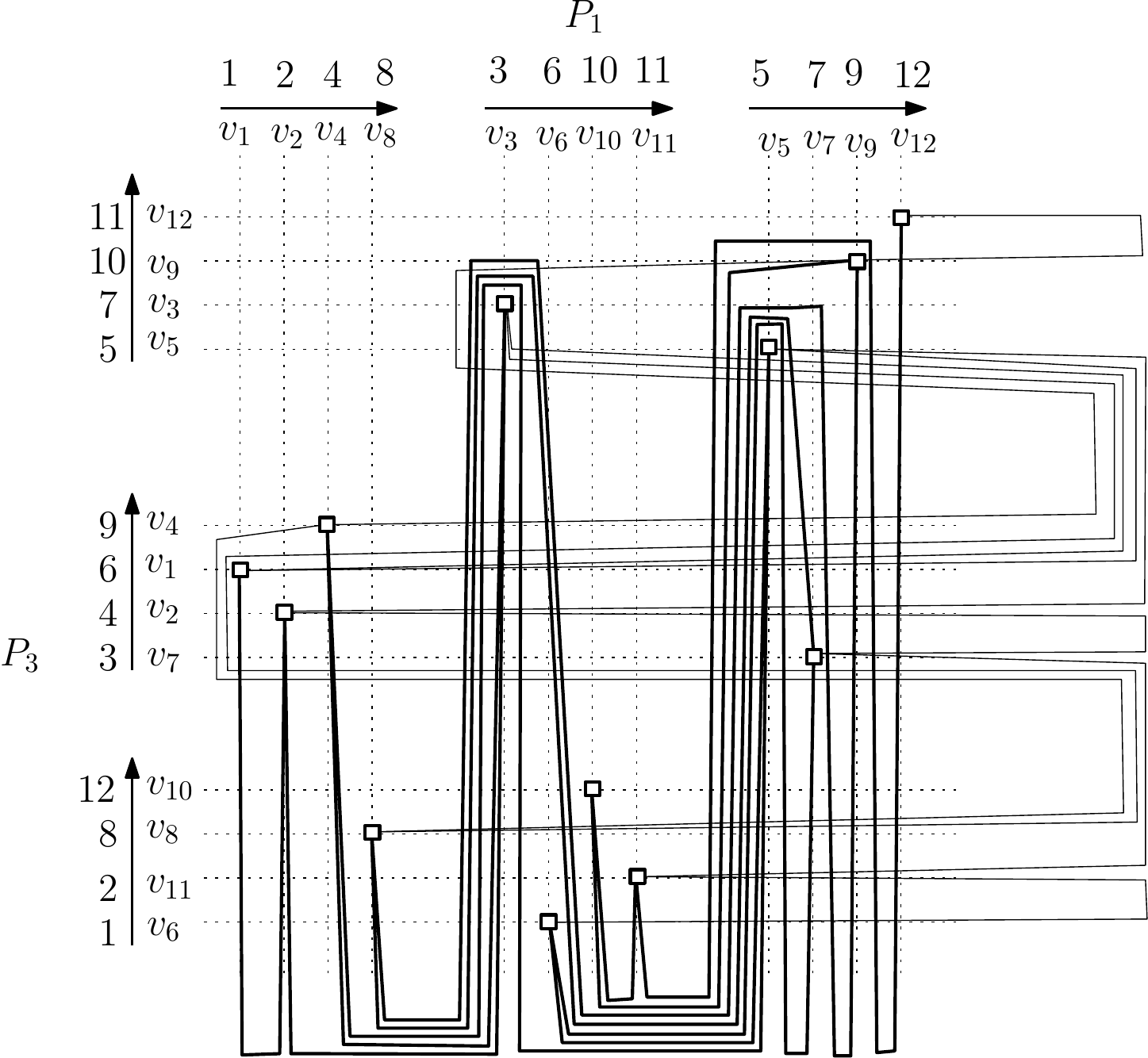}
\caption{Drawings of $P_1$ and $P_3$ on the point set of Figure~\ref{fig:1}(c).
}
\label{fig:2}
\end{figure}
\begin{figure} [pt]
\centering
\includegraphics[width=\textwidth]{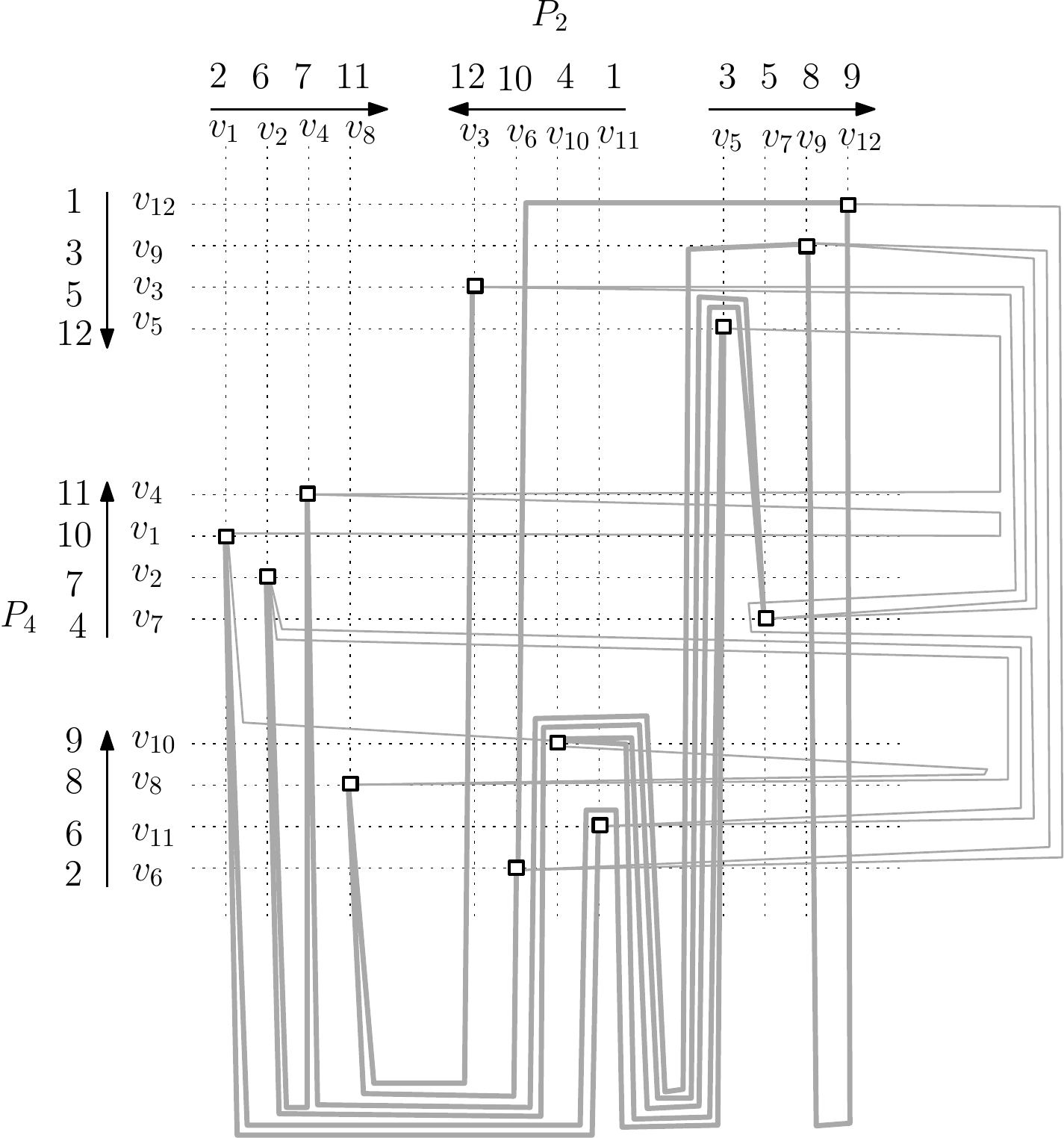}
\caption{Drawings of $P_2$ and $P_4$ on the point set of Figure~\ref{fig:1}(c).
}
\label{fig:3}
\end{figure}

\end{document}